\documentclass[a4paper,12pt,justification=centering]{article}
\usepackage{latexsym}
\usepackage{rotating}
\usepackage{amsmath,amssymb,amsxtra,amscd,amsthm}
\usepackage{tikz}
\usetikzlibrary{calc,decorations.markings}
\usepackage[mathscr]{eucal}
\usepackage{graphicx}
\usepackage[justification=centering]{caption}
\usepackage{subfig}
\usepackage{datetime}
\usepackage{pdfsync}
\usepackage{thumbpdf}
\usepackage{color}
\ifx\pdfoutput\undefined

\usepackage{amsmath}
\usepackage{amssymb}
\usepackage{pgfplots}
\usepackage{graphicx} 
\usepackage{caption}
\usepackage{ mathrsfs }
\usepackage{subcaption}
\usepackage{color}
\usepackage[nottoc,notlot,notlof]{tocbibind}

\usepackage{verbatim}
\usepackage{tensor}
\usepackage{xcolor}
\usepackage{titlesec}

\graphicspath {{images/}}
\usepackage{wrapfig}

\usepackage
[dvips,
 verbose=false,
 bookmarks=true,
 colorlinks=true,
 linkcolor=webred,
 filecolor=webbrown,
 citecolor=webgreen,
 pagecolor=webblue,
 urlcolor=webblue,
 pdftitle={},
 pdfauthor={Murad Alim},
 pdfsubject={},
 pdfkeywords={},
 bookmarksopen=false,
 pdfpagemode=None,
 pdfview=FitH,
 pdfstartview=FitH,
 extension=pdf]{hyperref}
\else
\usepackage
[pdftex,
 verbose=false,
 bookmarks=true,
 colorlinks=true,
 linkcolor=webred,
 filecolor=webbrown,
 citecolor=webgreen,
 pagecolor=webblue,
 urlcolor=webblue,
 pdftitle={},
 pdfauthor={Murad Alim},
 pdfsubject={},
 pdfkeywords={},
 bookmarksopen=false,
 pdfpagemode=None,
 pdfview=FitH,
 pdfstartview=FitH,
 extension=pdf]{hyperref}
\fi
\definecolor{webred}{rgb}{.8,0,0}
\definecolor{webbrown}{rgb}{.6,0,0}
\definecolor{webgreen}{rgb}{0,0.5,0}
\definecolor{webdkgreen}{rgb}{0,0.3,0}
\definecolor{webblue}{rgb}{0,0,0.5}

\numberwithin{equation}{section}

\providecommand{\href}[2]{#2}

\allowdisplaybreaks
\usepackage{bbm}
\usepackage{rotating}
\newcommand{\be}{\begin{eqnarray}}
\newcommand{\beq}{\begin{eqnarray}}
\newcommand{\ee}{\end{eqnarray}}

\def\e#1\e{\begin{equation}#1\end{equation}}
\def\ea#1\ea{\begin{align}#1\end{align}}

\theoremstyle{plain}
\newtheorem{thm}{Theorem}[section]
\newtheorem*{thm*}{Theorem}

\newtheorem{prop}[thm]{Proposition}
\newtheorem{cor}[thm]{Corollary}
\theoremstyle{definition}

\newtheorem{rem}[thm]{Remark}

\begin{document}

\setlength{\parindent}{0cm}
\setlength{\baselineskip}{1.5em}
\title{\bf{Intrinsic non-perturbative topological strings}}

\author{Murad Alim\footnote{\tt{murad.alim@uni-hamburg.de}}\\
\small Department of Mathematics, University of Hamburg, Bundesstr. 55, 20146, Hamburg, Germany}

\date{}
\maketitle

\abstract{We study difference equations which are obtained from the asymptotic expansion of topological string theory on the deformed and the resolved conifold geometries as well as for topological string theory on arbitrary
families of Calabi-Yau manifolds near generic singularities at finite distance in the moduli space. Analytic solutions in the topological string coupling to these equations are found. The solutions are given by known special functions and can be used to extract the strong coupling expansion as well as the non-perturbative content. The strong coupling expansions show the characteristics of D-brane and NS5-brane contributions, this is illustrated for the quintic Calabi-Yau threefold. For the resolved conifold, an expression involving both the Gopakumar-Vafa resummation as well as the refined topological string in the Nekrasov-Shatashvili limit is obtained and compared to expected results in the literature. Furthermore, a precise relation between the non-perturbative partition function of topological strings and the generating function of non-commutative Donaldson-Thomas invariants is given. Moreover, the expansion of the topological string on the resolved conifold near its singular small volume locus is studied. Exact expressions for the leading singular term as well as the regular terms in this expansion are provided and proved. The constant term of this expansion turns out to be the known Gromov-Witten constant map contribution.
}

\clearpage


\tableofcontents

\section{Introduction}

Topological string theory bridges research in mathematics and physics and has been a rich source of insights for both areas. It provides a quantitative handle on physical dualities as well as exact computations, see e.~g.~\cite{Neitzke:2004ni}. Within the context of mirror symmetry, topological string theory provides the tools to study higher genus mirror symmetry \cite{Marinobook,coxkatz,Alimlectures}.

Physically, the appeal of studying topological string theory originates from the fact that it shares features of physical strings while at the same time possessing clearer mathematical structures which allow one to seek answers to difficult physical questions. One such aspect, which is the focus of this work, is the fact that the free energy of topological string theory is only defined perturbatively as an asymptotic series in the topological string coupling. This is also known to be true for physical string theories as well as for many quantum field theories, see \cite{Shenker:1990uf,Marinolecture} and references therein.

The topological string partition function of a given family of Calabi-Yau (CY) threefolds is defined perturbatively in the topological string coupling $\lambda$. The free energies at genus $g$ are related recursively to the free energies at lower genera by the holomorphic anomaly equations \cite{Bershadsky:1993ta,Bershadsky:1993cx}. The anomaly equations were identified in \cite{Witten:1993ed} as the projective flatness equations of a connection identifying Hilbert spaces obtained from polarization choices of a geometric quantization problem associated to the moduli space of the family. The topological string partition function becomes a flat section of this bundle of Hilbert spaces over the moduli space. Further details about this quantum mechanical interpretation including a possible Hamiltonian as well as a description of other states in this Hilbert space remain puzzling, see \cite{Neitzke:2007yw} for a the description of other potential states in this Hilbert space.

The OSV conjecture \cite{Ooguri:2004zv} relates the topological string wave-function to the partition function of black holes. It provides an expectation that the non-perturbative content of topological string theory should be captured by the partition function of objects which are defined non-perturbatively on the same Calabi-Yau family, namely the BPS states forming the black holes. A number of challenges, most sharply collected in \cite{Denef:2007vg}, however limit the applicability of this connection for the study of the non-perturbative structure of topological string theory. Mathematically, the perturbative definition of topological string theory on one side of mirror symmetry corresponds to the study Gromov-Witten (GW) theory, while the enumerative content of the BPS states forming the black holes is captured by Donaldson-Thomas (DT) invariants, the equivalence of the enumerative geometry content of the two is the MNOP conjecture \cite{MNOP1,MNOP2}. 

The study of the non-perturbative structure of topological string theory has been much more promising for non-compact CY manifolds which exhibit dualities with Chern-Simons theory and matrix models, see e.~g.~\cite{Marinobook}. This is also the context for which early all genus results were obtained including the all genus topological string theory on the deformed conifold, whose partition function correspond to the $c=1$ string as well as to the Gaussian matrix model \cite{Ghoshal:1995wm}. The relevance of this result is that it also provides the expected universal behavior of topological string theory near singularities in the moduli space where finitely many states of the corresponding effective theory are becoming massless. Another all-genus result for topological string theory comes from a large N duality with Chern Simons theory \cite{GV} and provides the perturbative expansion of topological string theory on the resolved conifold. The study of the non-perturbative structure of topological string theory in these two cases as well as other non-compact CY manifolds has benefited a lot from their relation to matrix models, see e.~g.~\cite{Pasquetti:2009jg,Marinolecture} and references therein.

For non-compact CY geometries, another path gives further insights into the non-perturbative structure of topological string theory. Considering mirror non-compact CY geometries whose relevant data is captured by the mirror curve, a quantum mechanical problem was put forward in \cite{ADKMV}. The curve equation in these cases is characterized by an algebraic equation in two complex variables which take values in $\mathbb{C}^*$ or $\mathbb{C}$. These variables are identified with conjugate phase space variables and the defining equation of the curve is interpreted as a Hamiltonian whose eigenstates are wave-functions. This quantum curve approach is useful for the study of the relation of topological string theory to integrable systems. The approach of \cite{ADKMV} was revisited in the context of refined topological string theory in \cite{Aganagic:2011mi}, shedding light on the relation of the quantization of integrable systems of supersymmetric theories and topological strings \cite{Nekrasov:2009rc}.

Building on the quantum curve developments as well as on a series of insights from the study of ABJM theory \cite{Aharony:2008ug} (see \cite{Marinoloc} and references therein), a proposal for the non-perturbative definition of topological string theory was put forward in \cite{Hatsuda:2013oxa} and further scrutinized in \cite{Grassi:2014zfa}, using spectral properties of the quantum mechanical problem defined by the quantum curve. Interestingly, the wave functions obtained in this way were related to the Nekrasov-Shatashvili (NS) limit \cite{Nekrasov:2009rc} of the refined topological string with the quantization parameter $\hbar\sim \frac{1}{\lambda}$ being related to the inverse of the topological string coupling rather than $\hbar \sim \lambda^2$ as is expected from the wave-function interpretation of the anomaly equation \cite{Witten:1993ed}, which suggests perhaps two dual quantum mechanical pictures for the topological string. The spectral properties of the quantum mechanical system allowed the authors of \cite{Hatsuda:2013oxa,Grassi:2014zfa} however to extract both expansions in $\lambda$ and $\frac{1}{\lambda}$. The quantum curve setting was also recently used in \cite{Coman:2018uwk,Coman:2020qgf} to propose non-perturbative partition functions for topological strings on local CY manifolds related to the class $\mathcal{S}$ theories of \cite{GMN}.

In the case of the all genus free energy of the deformed conifold, the Borel resummation gives the Barnes G-function and can be used to access the non-perturbative content of the partition function as well as the corresponding matrix model, see for instance \cite{Pasquetti:2009jg,Marinolecture}. A proposal for the non-perturbative structure of topological string theory on the resolved conifold was given in \cite{Lockhart:2012vp}, making use of correspondence with supersymmetric indices. In \cite{Hatsuda:2015owa} a modified Borel resummation was applied to the resolved conifold and the expected non-perturbative structure of \cite{Hatsuda:2013oxa,Hatsuda:2015oaa} was obtained. The expected non-perturbative structure of \cite{Hatsuda:2015owa} was further obtained in \cite{Krefl:2015vna} from the exact duality with Chern-Simons theory. 

Remarkable progress in defining non-perturbative topological string theory was achieved recently in mathematics \cite{BridgelandDT,BridgelandCon}, inspired by \cite{Gaiotto:2014bza}. In \cite{BridgelandDT} a Riemann-Hilbert (RH) problem was put forward which describes the wall-crossing phenomena in Donaldson-Thomas theory, This corresponds physically to the wall-crossing phenomena of BPS states whose recent study has been advanced by \cite{GMN} and many others. In \cite{BridgelandDT}, the solution of the RH problem for the Argyres-Douglas $A_1$ theory was given, this corresponds on the topological string side to the deformed conifold free energy.\footnote{The observation of the link to the deformed conifold has not been made in \cite{BridgelandDT}, but is perhaps obvious to the experts.} The subsequent \cite{BridgelandCon} solves the RH problem for the resolved conifold and suggests the resulting Tau function as a non-perturbative definition of the topological string theory on the resolved conifold given its analytic properties and the fact that it contains as an asymptotic expansion the Gromov-Witten theory of the resolved conifold. In a sense, the work of Bridgeland provides several missing links in the expectation that the BPS content of a given geometry provides the non-perturbative definition of topological string theory on that geometry. The details of this program however become quickly very challenging since it requires as an input the complete relevant BPS spectra and their wall-crossing behavior. This seems currently, especially for compact geometries, very challenging if not intractable.

It is natural to wonder whether there is a more intrinsic path towards non-perturbative topological string theory which does not rely on dualities to other physical or mathematical problems and is as such not limited in its scope of applicability. Given the asymptotic nature of the expansion of the free energies, the Borel resummation of the free energy as well as the application of resurgence techniques are such paths, see e.~.g.~\cite{Aniceto:2011nu,Santamaria:2013rua,Couso-Santamaria:2014iia} as well as \cite{Couso-Santamaria:2016vwq} for a matching of the resurgence results with \cite{Grassi:2014zfa}. In the case of asymptotic series stemming from differential equations with irregular singular points, the knowledge of the differential equation itself is often more powerful than the knowledge of the asymptotic expansion around singular points. Especially in problems of mathematical physics, the ODEs in question are often the ones which have been well-studied for a long time. Such a differential equation in the topological string coupling is however not part of the defining data of topological strings. The quest for such differential equation in the string coupling was the motivation for \cite{Alim:2015qma}. In that paper, the holomorphic anomaly equations \cite{Bershadsky:1993cx} as well as the polynomial structure of the higher genus topological string amplitudes \cite{Yamaguchi:2004bt,Alim:2007qj} were used to obtain a differential equation in the string coupling in a certain scaling limit. The relevant differential equation turned out to be the Airy equation. Apart from the expected asymptotic expansion in this limit the equation has a solution which is non-perturbative in the string coupling. The latter was subsequently related to non-perturbative resurgence effects of NS-branes \cite{Couso-Santamaria:2015hva}.

The aim of this work is to extend the study of the intrinsic characterization of the non-perturbative structure of topological string theory. We are in a fortunate situation where many pieces of the puzzle are already available in the recent physics and especially mathematics literature and can be readily used and put together. The starting point is a difference equation which was first proved in \cite{Iwaki} for the free energies of the WKB analysis of the Weber curve. This curve is related to the deformed conifold. A similar difference equation was proved in \cite{alim2020difference} for the free energies of the resolved conifold. Both derivations only have the asymptotic expansion as their input. A first aim of this work is to use the expected universal behavior of topological string theory on arbitrary families of CY threefolds near singular loci in the moduli space where finitely many states of the effective theory become massless and derive a difference equation for the topological string free energies in a limit around these loci. We next identify the Barnes G-function as a solution for the difference equations of the deformed conifold as well as for the universal behavior near the singularities.  A solution of the difference equation for the resolved conifold \cite{alim2020difference} was identified in \cite{alim2021integrable} using building blocks of Bridgeland's Tau function for the resolved conifold \cite{BridgelandCon}. The explicit analytic solutions can be used to obtain the strong coupling expansions of topological string theory as well as to express their non-perturbative content.  The characteristic traits of non-perturbative effects due to D-branes and NS-branes are obtained. Moreover, for the resolved conifold an expression involving both the Gopakumar-Vafa resummation as well as the refined topological string in the Nekrasov-Shatashvili limit is obtained. The latter was put forward in \cite{Hatsuda:2013oxa,Hatsuda:2015owa}, we obtain a matching with their results up to some factors which are discussed.

The organization of this work is as follows. In sec.~\ref{sec:freeenergies}, the topological string free energies are recalled as well as their Gromov-Witten and Gopakumar-Vafa expansions. We proceed with a discussion of the expected universal behavior of topological string theory near singularities where finitely many states of the underlying effective theory in $4d$ become massless. The explicit expressions of the topological string free energies for the deformed and resolved conifold are given. In sec.~\ref{sec:diffeq}, the difference equations for the deformed and resolved conifold geometries are introduced and a similar equation for the universal behavior of the free energies in a limit around singular points is derived. We proceed with discussing the analytic solutions in the string coupling of the difference equations and extract their strong coupling expansion as well as the non-perturbative content in sec.~\ref{sec:nonpertcontent}. We furthermore give an exact non-perturbative relation between the topological string partition function and the generating function of non-commutative DT invariants. Finally we study in detail the expansion of the free energies of the resolved conifold near the locus where the $\mathbb{P}^1$ of the resolution shrinks to zero and the corresponding coordinate $t\rightarrow 0$. We prove an exact expression of the leading singular behavior as well as the sub-leading terms. In particular the constant terms in this expansion turn out to be the contributions of constant maps in Gromov-Witten theory, the higher order terms are polynomials in the coordinate. Moreover, this provides a mathematical proof in this case of the \emph{gap condition} which is expected on physical grounds and was used in \cite{Huang:2006si,Huang:2006hq} in the study of higher genus mirror symmetry. We finish in sec.~\ref{sec:conclusions} with the conclusions.


\section{Topological string free energies}\label{sec:freeenergies}

To a mirror family of CY threefolds, topological string theory associates the topological string partition function which is defined as an asymptotic series in the topological string coupling $\lambda$, summing over the free energies associated to the world-sheets of genus $g$:
\begin{equation}
Z_{top} (\lambda,t)= \exp \left(\sum_{g=0}^{\infty} \lambda^{2g-2} \mathcal{F}^{g}(t)\right)\,.
\end{equation}
Where $t=(t^1,\dots,t^n)$ is a set of distinguished local coordinates on the underlying moduli space $\mathcal{M}$, which is of dim $n=h^{1,1}(X_t)=h^{2,1}(\check{X}_{t(z}))$. $X_t$ and $\check{X}_{t(z)}$ are a mirror pair of CY threefolds which correspond to the A-model and B-model sides of mirror symmetry. The map $t(z)$ on the B-model side expressing the distinguished coordinates in terms of the more natural complex structure coordinates $z$ is the mirror map. It is useful to consider the total space of a line bundle $\mathcal{L}\rightarrow \mathcal{M}$ whose sections correspond to a distinguished vacuum state in the underlying SCFT and which has a different geometric interpretation on both sides of mirror symmetry. $\mathcal{M}$ is a projective special K\"ahler manifold. The special geometry as well as the holomorphic anomaly equations of BCOV, together with the boundary conditions of sec.~\ref{sec:boundary} can be used to geometrically characterize the topological string free energies at each genus. The latter are in particular non-holomorphic sections of $\mathcal{L}^{2-2g}$ \cite{Bershadsky:1993cx}. A holomorphic limit can be considered by taking the base point on $\mathcal{M}$ to $i\infty$ and expanding in canonical coordinates. 

\subsection{GW and GV expansions}
In the holomorphic limit together with an expansion around a distinguished large volume point in the moduli space, the topological string free energies become the generating functions of higher genus Gromov-Witten invariants on the A-model side of mirror symmetry. The GW potential of $X$ is the following formal power series:
\begin{equation}
F(\lambda,t) = \sum_{g\ge 0}  \lambda^{2g-2} F^g(t)= \sum_{g\ge 0}  \lambda^{2g-2} \sum_{\beta\in H_2(X,\mathbb{Z})}  N^g_{\beta} \,q^{\beta}\, ,
\end{equation}
where $q^{\beta} := \exp (2\pi i \langle t,\beta \rangle)$ is a formal variable living in a suitable completion of the effective cone in the group ring of $H_2(X,\mathbb{Z})$.

The GW potential can be furthermore written as:
\begin{equation}
F=F_{\beta=0} + \tilde{F}\,,
\end{equation}
where $F_{\beta=0}$ denotes the contribution from constant maps and $ \tilde{F}$ the contribution from non-constant maps. The constant map contribution at genus 0 and 1 are $t$ dependent and the higher genus constant map contributions take the universal form \cite{Faber}:
\begin{equation}
F_{\beta=0}^g = \frac{\chi(X)(-1)^{g-1}\, B_{2g}\, B_{2g-2}}{4g (2g-2)\, (2g-2)!}\,, \quad g\ge2\,,
\end{equation}
where $\chi(X)$ is the Euler characteristic of $X$ and the Bernoulli numbers $B_n$ are generated by:
\begin{equation}
\frac{w}{e^w-1} = \sum_{n=0}^{\infty} B_n \frac{w^n}{n!}\,.
\end{equation}

The Gopakumar-Vafa (GV) resummation of the GW potential \cite{Gopakumar:1998ii,Gopakumar:1998jq} reformulates the non-constant part of the GW potential in terms of the Gopakumar-Vafa invariants  $n^g_{\beta} \in \mathbb{Z}$ which are given by a count of electrically charged $M_2$ branes in an M-theory setup. The GW potential can thus be written as:
\begin{equation}\label{GVresum}
\tilde{F}(\lambda,t)= \sum_{\beta>0}\sum_{g\ge 0} n^g_{\beta}\, \sum_{k\ge 1} \frac{1}{k} \left( 2 \sin \left( \frac{k\lambda}{2}\right)\right)^{2g-2} q^{k\beta}\,.
\end{equation}

in particular
$$ \tilde{F}^0(t)=\sum_{\beta>0} n_{\beta}^0 \textrm{Li}_3(q^{\beta})\,, \quad q^{\beta}= \exp(2\pi i t^{\beta})\,.$$ 


\subsection{Behavior near singularities}\label{sec:boundary}
Fixing a frame for $\mathcal{L}$ we will denote the functions, which are obtained from $\mathcal{F}^g \in \mathcal{L}^{2g-2}$ by $F^g$. The leading singular behavior of the free energy $F^g$ at a conifold locus has been determined in \cite{Bershadsky:1993ta,Bershadsky:1993cx,Ghoshal:1995wm,Antoniadis:1995zn,Gopakumar:1998ii,Gopakumar:1998jq}
\begin{equation} \label{Gap}
 F^g(t_c)= b \frac{B_{2g}}{2g (2g-2) t_c^{2g-2}} + O(t^0_c),
\qquad g>1\,.
\end{equation}
Here $t_c\sim \Delta^{\frac{1}{m}}$ is the special coordinate at the discriminant locus $\Delta=0$.\footnote{$\Delta$ is usually defined using the algebraic moduli of the problem, see e.~g.~appendix \ref{sec:quintic}.} For a conifold singularity $b=1$ and $m=1$. In particular the leading singularity in \eqref{Gap} as well as the absence of subleading singular terms follows from the Schwinger loop computation of \cite{Gopakumar:1998ii,Gopakumar:1998jq}, which computes the effect of the extra massless hypermultiplet
in the space-time theory \cite{Vafa:1995ta}. The singular structure and the ``gap''
of subleading singular terms have been also observed in the dual matrix model
\cite{Aganagic:2002wv} and were first used in \cite{Huang:2006si,Huang:2006hq}
to fix the holomorphic ambiguity at higher genus. The space-time derivation of \cite{Gopakumar:1998ii,Gopakumar:1998jq} is
not restricted to the conifold case and applies also to the case $m>1$
singularities which give rise to a different spectrum of
extra massless vector and hypermultiplets in space-time.
The coefficient of the Schwinger loop integral is a weighted trace over the spin of the particles~\cite{Vafa:1995ta, Antoniadis:1995zn} leading to the prediction $b=n_H-n_V$ for the coefficient of the leading singular term. A higher genus result with a singularity leading to  $b=1-1=0$ was studied in \cite{Alim:2008kp}, following the effective field theory analysis of \cite{Klemm:1996kv}.


\subsection{Deformed and resolved conifolds}

The conifold singularity refers to a singular point in a threefold that locally looks like
\begin{equation}
(x_1\, x_2 - x_3 \,x_4=0) \subset \mathbb{C}^4\,,
\end{equation}
the singularity can be deformed by introducing a parameter $a\in\mathbb{C}^*$, after changes of the local coordinates in $\mathbb{C}^4$ this can be brought to the form:
\begin{equation}
X_{a}= \left\{ x_1\, x_2 + y^2= x^2 - 4 a  \right\}
\end{equation}
which defines the family of non-compact CY threefolds known as the deformed conifold. Note that the addition of a complex parameter in the defining equation amounts to a change of the complex structure so it is natural to study this geometry using the B-model topological string theory. The result of \cite{Ghoshal:1995wm} is that the topological string free energy on this geometry obtained from the relation to the $c=1$ string has the form\footnote{Compared to \cite{Ghoshal:1995wm} we have added the $\lambda$ dependence as well as the $-\frac{3}{4}a^2$ in order to match more recent results, such as \cite{Iwaki}.}:
\begin{equation}\label{defconfree}
F(\lambda,a)= \lambda^{-2 } \,  \left( \frac{a^2}{2} \log a -\frac{3}{4} a^2\right) -\frac{1}{12} \log a +  \sum_{g=2}^{\infty} \frac{B_{2g}}{2g(2g-2) a^{2g-2}} \lambda^{2g-2}\,.
\end{equation}
We note that in the quantum curve setting, the curve which is quantized is: 
\begin{equation}
\Sigma_a:= \left\{  y^2=x^2-4a\subset \mathbb{C}^2\right\}\,,
\end{equation}
and corresponds to the harmonic oscillator. In the exact WKB setting, this curve is known as the Weber curve and is discussed, e.~g.~in \cite{Iwaki}. Furthermore the WKB analysis encodes the BPS content of the corresponding effective $4d,\mathcal{N}=2$ theory obtained from compactifying type IIB string theory on this non-compact CY. In this case this gives the Argyres-Douglas $A_1$ theory.

The CY threefold given by the total space of the rank two bundle over the projective line:
\begin{equation}
X_t := \mathcal{O}(-1) \oplus \mathcal{O}(-1) \rightarrow \mathbb{P}^1\,,
\end{equation}
corresponds to the resolution of the conifold singularity in $\mathbb{C}^4$ and is known as the resolved conifold. This geometry is defined on the A-model side of mirror symmetry and $t$ corresponds to:
\begin{equation}
t= \int_{C} B+ i \omega\,,
\end{equation}
where $B\in H^{2}(X,\mathbb{R})/H^{2}(X,\mathbb{Z})$ is the B-field, $\omega$ is the K\"ahler form and $C$ corresponds to the $\mathbb{P}^1$ class in this geometry. The GW potential for this geometry was determined in physics \cite{Gopakumar:1998ii,GV}, and in mathematics \cite{Faber} with the following outcome for the non-constant maps:\footnote{See also \cite{MM} for the determination of $F^g$ from a string theory duality and the explicit appearance of the polylogarithm expressions.}
\begin{equation}\label{resconfree}
\tilde{F}(\lambda,t)= \sum_{g=0}^{\infty} \lambda^{2g-2} \tilde{F}^g(t)= \frac{1}{\lambda^2} \textrm{Li}_{3}(q)+ \sum_{g=1}^{\infty} \lambda^{2g-2} \frac{(-1)^{g-1}B_{2g}}{2g (2g-2)!} \textrm{Li}_{3-2g} (q) \, ,
\end{equation}
where $q:=\exp(2\pi i \,t)$ and the polylogarithm ist defined by:
\begin{equation}
\textrm{Li}_s(z) = \sum_{n=0}^{\infty} \frac{z^n}{n^s}\, ,\quad s\in \mathbb{C}\,.
\end{equation}



\section{Difference equations and their solutions}\label{sec:diffeq}
\subsection{Difference equations}
In the following we will review the derivation of a difference equation which was obtained in \cite{Iwaki} for the free energies of the Weber curve which correspond to the free energies of the deformed conifold geometry and which was adapted in \cite{alim2020difference} for the free energies of the resolved conifold.  

\begin{thm} \cite{Iwaki,alim2020difference} \label{diffeq} The free energy of the deformed conifold \cite{Iwaki} satisfies the following difference equation:
\begin{equation}
F(\lambda,a+\lambda) - 2 F(\lambda,a) + F(\lambda,a-\lambda)= \frac{\partial^2}{\partial a^2}\, F^0(a) \,,\end{equation}
with:
\begin{equation}
F^0(a)= \frac{1}{2} a^2 \log a -\frac{3}{4} a^2\,,
\end{equation}
and the free energy of the resolved conifold satisfies \cite{alim2020difference}:
\begin{equation}
\tilde{F}(\lambda,t+\check{\lambda}) - 2 \tilde{F}(\lambda,t) + \tilde{F}(\lambda,t-\check{\lambda})= \left(\frac{1}{2\pi }\frac{\partial}{\partial t}\right)^2\, \tilde{F}^0(t) \,,\quad \check{\lambda}=\frac{\lambda}{2\pi}\,.
\end{equation}
with
\begin{equation}
\quad \tilde{F}_{0}(t)= \textrm{Li}_{3}(q)\,.
\end{equation}
\end{thm}

The two versions of the theorem were proved in \cite{Iwaki,Iwaki2,alim2020difference}. The proof is included in the appendix \ref{sec:proof}, a notable feature is that it only requires the asymptotic expansion. 

For the discussion of the universal structure of topological strings near finite distance singularities, the following corollary of the above theorem is obtained:

\begin{cor}\label{diffarb}
For $X_t$ and $\check{X}_{t(z)}$, a mirror pair of CY threefolds which corresponding to the $A-$ and $B-$sides of mirror symmetry and which can be thought of as the fibers of corresponding families of over a base manifold $\mathcal{M}$ with  dim $\mathcal{M}=n$, where $n=h^{1,1}(X_t)=h^{2,1}(\check{X}_{t(z}))$.  We consider local coordinates $t=\left\{t^1,\dots,t^n\right\}$.  Let $t_c$ be a coordinate near a singularity of finite distance in the special K\"ahler metric, We assume the following behavior of  the topological string free energies, motivated by physical expectations:\footnote{In the following, w.~l.~o.~g.~we only highlight the dependence on the coordinate which corresponds to the singularity, for a study of the behavior of toplogical strings on higher dimensional moduli spaces near singularities in a compact setting see, e.~g.~\cite{Haghighat:2009nr,Alim:2012ss}.}
\begin{equation}
F^0(t_c)= b\left( \frac{1}{2} t_c^2 \log t_c -\frac{3}{4} t_c^2 \right)+ \mathcal{O}(t_c^3)  \,,\quad F^1(t_c)= -\frac{b}{12} \log t_c + \mathcal{O}(t_c)\,,
\end{equation}
and 
\begin{equation}
 F^{g}(t_c)= b \frac{B_{2g}}{2g (2g-2) t_c^{2g-2}} + O(t^0), \quad g\ge 2\,,
\end{equation}
consider now $\Lambda \in \mathbb{C}^*$ and the rescaling:
$$ \lambda' = \lambda \cdot \Lambda\,, \quad  t_c' = t_c \cdot \Lambda\,,$$
and define:
$$ F'(\lambda',t_c') := \lim_{\Lambda \rightarrow \infty} \left( F(\lambda',t_c') + b\frac{(t_c')^2}{2}\log \Lambda -\frac{b}{12}\log \Lambda\right)\,,$$
then the following difference equation is satisfied by $F'$:
\begin{equation}\label{diffeqarb}
F'(\lambda',t^{\circ},t_c'+\lambda') - 2 F'(\lambda',t^{\circ},t'_c) + F'(\lambda',t^{\circ},t_c'-\lambda') = b\,  \log t'_c \,.
\end{equation}
\end{cor}
\begin{proof}
First consider the all-genus topological string free energy near $t_c\rightarrow 0$ whose behavior is given by the assumptions:
\begin{equation}
\begin{split}
F(\lambda,t_c)&= \sum_{g=0}^{\infty} \lambda^{2g-2} F^g(t_c)= \frac{1}{\lambda^2}\left( b\left( \frac{1}{2} t_c^2 \log t_c -\frac{3}{4} t_c^2 \right)+ \mathcal{O}(t_c^3)\right) \\
&-\frac{b}{12} \log t_c + \mathcal{O}(t_c) + b \sum_{g=2}^{\infty} \lambda^{2g-2} \left(\frac{B_{2g}}{2g (2g-2) t_c^{2g-2}} + \mathcal{O}(t_c^0)\right)\,.
\end{split}
\end{equation}
Inserting the rescaled $t_c'$ and $\lambda'$ we obtain:
\begin{equation}
\begin{split}
F(\lambda',t'_c)&= \sum_{g=0}^{\infty} (\lambda')^{2g-2} F^g(t'_c)= \frac{1}{(\lambda')^2}  b\left( \frac{1}{2} (t'_c)^2 \log t'_c -\frac{3}{4} (t'_c)^2 \right)- b\frac{(t'_c)^2}{2}\log \Lambda +  \mathcal{O}(1/\Lambda)\\
& -\frac{b}{12} \log t'_c +\frac{b}{12}\log \Lambda +  \mathcal{O}(1/\Lambda) + b \sum_{g=2}^{\infty} (\lambda')^{2g-2} \left(\frac{B_{2g}}{2g (2g-2) (t'_c)^{2g-2}} + \mathcal{O}(1/\Lambda^{2g-2})\right)\,.
\end{split}
\end{equation}
Hence we obtain for:
\begin{equation}
F'(\lambda',t_c') := \lim_{\Lambda \rightarrow \infty} \left( F(\lambda',t_c') + b\frac{(t_c')^2}{2}\log \Lambda -\frac{b}{12}\log \Lambda\right) = b\cdot F_{\textrm{def}}(\lambda',t'_c)\,,
\end{equation}
where $F_{\textrm{def}}$ is the free energy of the deformed conifold \ref{defconfree}. The proof of the corollary therefore proceeds as in the latter case which is given in the appendix.
\end{proof}

\subsection{Solution for the deformed conifold}
We proceed with the discussion of the solutions of the difference equations. We begin by introducing the Barnes G-function, which is defined by:
\begin{equation}\label{BarnesG}
\begin{split}
G(z+1)&= \Gamma(z) \, G(z)\,, \quad z\in \mathbb{C}\,, \\
G(1)&=G(2)=G(3)=1\, ,\quad \frac{d^3}{dz^3} \log G(z) \ge 0\,, \quad z>0.
\end{split}
\end{equation}

One of the equivalent forms to express the G-function is the Weierstrass canonical product, see \cite{Adamchik}:
\begin{equation}
G(z+1)= (2\pi)^{\frac{z}{2}} \exp \left(  - \frac{z+z^2(1+\gamma)}{2}\right) \, \prod_{k=1}^{\infty} \left( 1+ \frac{z}{k} \right)^k \exp \left(\frac{z^2}{2k} -z  \right)\,,
\end{equation}
where $\gamma$ is the Euler-Mascheroni constant. The logarithm Barnes G-function has moreover the following Taylor expansion around $z=0$:
\begin{equation}
\log G(z+1)= \frac{1}{2} \left( \log 2\pi -1\right) z - (1+\gamma) \frac{z^2}{2} +\sum_{n=3}^{\infty} (-1)^{n-1} \zeta(n-1) \frac{z^n}{n}\,,
\end{equation}
as well as asymptotic expansion for $z\rightarrow \infty$.
\begin{equation}
\log G(z+1) = \frac{z^2}{2} \left(  \log z - \frac{3}{2}\right) -\frac{1}{12} \log z - z \zeta'(0) + \zeta'(-1)- \sum_{g=2}^{\infty} \frac{B_{2g}}{2g(2g-2) z^{2g-2}}\,,
\end{equation}
where $\zeta$ is the $\zeta-$function. We define:
\begin{equation}\label{defconnonp}
F_{np}(\lambda,a):= \log G\left(1+\frac{a}{\lambda}\right) +\frac{a^2}{2\lambda^2} \log \lambda+ \frac{a}{\lambda} \zeta'(0) + \frac{1}{12}\log(\lambda)+\zeta'(-1)\,,
\end{equation}
and obtain the following:
\begin{prop} 
$F_{np}(\lambda,a) $ is the unique solution to the difference equation for the free energies of the deformed conifold with asymptotic behavior fixed by \eqref{defconfree}.
\end{prop}
\begin{proof}
Using the functional equation of the Barnes $G$-function we obtain:
\begin{equation}
\begin{split}
&\log G\left(1+\frac{a+\lambda}{\lambda}\right)+ \log G\left(1+\frac{a-\lambda}{\lambda}\right) - 2 \log G\left(1+\frac{a}{\lambda}\right) \\
&= \log \Gamma\left(\frac{a}{\lambda}+1\right) - \log \Gamma \left( \frac{a}{\lambda}\right) = \log \frac{a}{\lambda}\,.
\end{split}
\end{equation}
From the additional terms in \eqref{defconnonp}, only the quadratic term in $a$ contributes to the r.h.s. of the difference equation. We obtain:
\begin{equation}
F_{np}(\lambda,a+\lambda) - 2 F_{np}(\lambda,a) + F_{np}(\lambda,a-\lambda)= \log a = \frac{\partial^2}{\partial a^2}\, F^0(a)\,.
\end{equation}
For a proof of the uniqueness of the results one may follow exactly the same reasoning as in \cite{alim2021integrable}.
\end{proof}

\subsection{Solution for the resolved conifold}
The solution of the difference equation for the resolved conifold was identified in \cite{alim2021integrable}, by adapting building blocks of Bridgeland's Tau function for the resolved conifold \cite{BridgelandCon}. The special functions in \cite{BridgelandCon} involve the multiple sine functions which are defined using the Barnes multiple Gamma functions \cite{Barnes}. For a variable $z\in \mathbb{C}$ and parameters $\omega_1,\ldots,\omega_r \in \mathbb{C}^{*}$ these are defined by:
\begin{equation}
    \sin_r(z\,|\, \omega_1,\dots,\omega_r):= \Gamma_{r}(z\, |\, \omega_1,\dots,\omega_r) \cdot \Gamma_{r}\left(\sum_{i=1}^r \omega_i - z\, |\, \omega_1,\dots,\omega_r\right)^{(-1)^r} \,,
\end{equation}
for further definitions, see e.~g.~ \cite{BridgelandCon,Ruijsenaars1} and references therein. We introduce furthermore  the generalized Bernoulli polynomials, defined by the generating function:
\begin{equation}
    \frac{x^r\, e^{zx}}{ \prod_{i=1}^r (e^{\omega_i x}-1)} = \sum_{n=0}^{\infty} \frac{x^n}{n!} \, B_{r,n}(z\,|\, \omega_1,\,\dots,\omega_r)\,.
\end{equation}

Consider now the function $G_3(z\,|\,\omega_1,\omega_2)$ of \cite[Sec. 4.2]{BridgelandCon}, defined by:\footnote{A subscript $3$ is added here to $G_3$ compared to \cite{BridgelandCon} to avoid confusion with the Barnes G-function.}
\begin{equation}\label{g3def}
    G_3(z\, | \, \omega_1,\omega_2) := \exp\left(\frac{\pi i}{6} \cdot B_{3,3}(z+\omega_1\,|\,\omega_1,\omega_1,\omega_2)\right) \cdot \sin_3(z+\omega_1\, |\, \omega_1,\omega_2,\omega_3),
\end{equation}
and define a function
\begin{equation}\label{resconfreedef}
    F_{np}(\lambda,t):= \log G_3(t\,|\,\check{\lambda},1)\,.
\end{equation}
It was shown in \cite{alim2021integrable} that $F_{np}$ is the unique solution of the difference equation with the boundary condition:
$$ \lim_{\lambda \rightarrow 0} \lambda^2  F_{np}(\lambda,t) = \tilde F^0(t)=\mathrm{Li}_3(e^{2\pi i t})\,,$$
and that moreover it has the following asymptotic expansion \cite{BridgelandCon,alim2021integrable}:
\begin{equation}
 F_{np} (\lambda,t) \sim \sum_{g=0}^{\infty} \lambda^{2g-2} \tilde{F}^{g}(t)\,,
\end{equation}
where $\tilde{F}^g(t)$ are the non-constant parts of the conifold free energies defined in \eqref{resconfree}.

\section{Non-perturbative content of the solutions}\label{sec:nonpertcontent}

\subsection{Deformed conifold}
We start by discussing the non-perturbative content of topological string theory on the deformed conifold, whose structure is universal for topological strings near finite distance singularities as discussed in sec.~\ref{sec:boundary}. We therefore consider the Taylor series expansion of the solution of the difference equation as $\lambda\rightarrow\infty$ using the Taylor series expansion of the Barnes G-function, we obtain:\footnote{We have used $\zeta'(0)=-\frac{1}{2}\log 2\pi$.}
\begin{equation}
\begin{split}
Z_{np}(\lambda,t)&=\exp(F_{np}(\lambda,t))\\
&= \exp\left(\zeta'(-1)+\frac{1}{12}\log(\lambda)- \frac{1}{2}  \frac{a}{\lambda} - (1+\gamma+\log \lambda) \frac{a^2}{2 \lambda^2} \right) \\
& \times \exp \left( \sum_{n=3}^{\infty} (-1)^{n-1} \zeta(n-1) \frac{a^n}{\lambda^n n} \right) \,.
\end{split}
\end{equation}
It is expected that non-perturbative effects due to D-branes are signaled by a factor of $e^{-1/g_s}$ and effects due to NS5-branes come with factors of $e^{-1/g_s^2}$ \cite{Strominger:1990et,Callan:1991dj,Becker:1995kb}, see also \cite{Alexandrov:2011va} and references therein as well as \cite{Couso-Santamaria:2015hva} for a discussion in the resurgence and topological string context, where $g_s$ refers to the physical string theory coupling under consideration. It is perhaps reasonable to expect a similar structure in topological string theory using the topological string coupling $\lambda$. The discussion of which solitonic branes correspond to the factors requires a distinction between the A- and B-model which would see non-perturbative objects of type IIA or IIB string theory respectively. Recall that the result of the free energies of the deformed conifold at hand is a B-model result. Although its all-genus structure is very clear, the corresponding A-model geometry is perhaps less clear. It would be interesting however to understand the precise connection between the non-perturbative topological string theory content in this case and the metric obtained in \cite{Ooguri:1996me} by smoothening the conifold singularity through instanton effects. Such as smoothening was expected in \cite{Strominger:1995cz}.

We proceed here with a discussion of the non-perturbative content of the strong coupling expansion in the case of the explicit example of the quintic and its mirror in the next subsection.

\subsection{Topological string universality and the quintic example}
By the corollary \ref{diffarb}, the same difference equation as for the deformed conifold also holds for topological string theory on an arbitrary family of Calabi-Yau manifolds near a locus where finitely many states of the effective field theory become massless. To exemplify this we consider the quintic Calabi-Yau threefold, whose definition and mirror construction are reviewed in the appendix \ref{sec:quintic}. The singular conifold locus of the quintic and its mirror has been studied in many works, starting with \cite{Candelas:1990rm}. The expectation of the physical behavior of topological string theory near this singularity was in particular used in \cite{Huang:2006hq} to supplement the polynomial solution \cite{Yamaguchi:2004bt} of the holomorphic anomaly equations \cite{Bershadsky:1993cx} with boundary conditions.

In this case the good special coordinate in the B-model side of mirror symmetry is given by a ratio of two solutions of the Picard-Fuchs equation \eqref{PF}, when the latter is considered in the coordinate:
$$ \delta = \frac{1-3125z}{3125 z}\,,$$
the outcome of \cite{Huang:2006hq} is that:
\begin{equation}
t_c(\delta)= \delta -\frac{3}{10} \delta^2 + \mathcal{O}(\delta^3)\,.
\end{equation}
This special coordinate is interpreted as the mirror map and thus gives the mass of the object on the A-model side, which becomes massless in this limit. The object becoming massless in this case has D6 brane charge which was obtained by the analytic continuation studied in \cite{Candelas:1990rm} and revisited in \cite{Huang:2006hq} and more recently in \cite{Knapp:2016rec}. By the physical reasoning reviewed in sec.~\ref{sec:boundary}, it was thus expected in \cite{Huang:2006hq} that:
$$ F^g(t_c) =\frac{B_{2g}}{2g (2g-2) t_c^{2g-2}} + O(t^0_c)\,, \quad g>1$$
the behavior for genus $0,1$ was further determined \cite{Huang:2006hq}:
$$ F^0(t_c)= \frac{1}{2}t_c^2 \log t_c - \frac{3}{4} t_c^2 + \mathcal{O}(t_c^3) \,, \quad F^{1}(t_c)=-\frac{1}{12} \log t_c + \mathcal{O}(t_c).$$
Consider now $\Lambda \in \mathbb{C}^*$ and the rescaling:
$$ \lambda' = \lambda \cdot \Lambda\,, \quad  t_c' = t_c \cdot \Lambda\,,$$
and define:
$$ F'(\lambda',t_c') := \lim_{\Lambda \rightarrow \infty} \left( F(\lambda',t_c') + \frac{(t_c')^2}{2}\log \Lambda -\frac{1}{12}\log \Lambda\right)\,,$$
then all the ingredients of corollary \ref{diffarb} are met, and by the proof of that corollary $F'$ satisfies the difference equation \ref{diffeqarb}. We can thus define the non-perturbative completion for topological string theory on the quintic in the limit $\Lambda\rightarrow \infty$ to be given by $F_{np}(\lambda',t_c')$, where $F_{np}$ is defined in \eqref{defconnonp}.

For the strong coupling expansion we thus obtain:
\begin{equation}
\begin{split}
Z_{np}(\lambda',t') &= \exp\left(\zeta'(-1)+\frac{1}{12}\log(\lambda')- \frac{1}{2}  \frac{t_c'}{\lambda'} - (1+\gamma+\log \lambda) \frac{(t_c')^2}{2 (\lambda')^2} \right) \\
& \times \exp \left( \sum_{n=3}^{\infty} (-1)^{n-1} \zeta(n-1) \frac{(t'_c)^n}{(\lambda')^n n} \right) \,,
\end{split}
\end{equation}
one may interpret the $-1/\lambda'$ coefficient $t_c'/2$ as a D-brane instanton action. The expected coefficient of the NS5-brane is the volume of the CY, to understand why the $(t_c')^2$ is still plausible as the volume in this limit we recall that the volume of the CY as determined by the special geometry at large radius is typically given by a period of the mirror, which has the form\footnote{See e.~g.~\cite{Alimlectures} and references therein.}:
$$ \mathcal{V}\sim 2F_0 - tF_t\,,$$
in the limit $t_c\rightarrow 0$ we have indeed:
$$ 2F_0(t_c) - t_c \,F_{t_c} = - \frac{t_c^2}{2}\,,$$
although geometric interpretations such as the volume become less clear once the large volume regime of the moduli space is left. It would be interesting to interpret the higher order terms in this expansion.

\subsection{The resolved conifold}

\subsubsection{Strong coupling expansion}
The strong coupling expansion for the non-perturbative free energy of the resolved conifold is obtained from the asymptotic expansion of $\log G_3$, which is given in \cite[Prop. 4.8]{BridgelandCon}, it is given by:
\begin{equation}
F_{np}(\lambda,t)= -\frac{\zeta(3)}{2\pi^2} \check{\lambda} - \frac{\pi i}{12} \left( t-\frac{1}{2}\right) + \frac{\pi i}{\check{\lambda}^2} \frac{1}{3!} \left( t^3-\frac{3}{2}t^2 + \frac{t}{2}\right) + \sum_{k=1}^{\infty} \frac{B_{2k}(t)\cdot B_{2k-2}}{2k! \cdot (2\pi i)} \, \left(\frac{2\pi i}{\check{\lambda}}\right)^{2k-1}\,,
\end{equation}
valid for for $\lambda\rightarrow \infty$ and $\textrm{Im}\, t>0$ and where the Bernoulli polynomials $B_n(t)$ are defined by the generating function:
\begin{equation}
\frac{x e^{tx}}{e^x-1} = \sum_{n=0}^{\infty} B_n(t) \frac{x^n}{n!}\,,
\end{equation} 
this expansion leads to the following:
\begin{rem}
\begin{itemize}
\item It is interesting to note that the asymptotic expansion for $\lambda\rightarrow 0$ of the non-perturbative free energy is naturally expressed in terms of a $q=\exp(2\pi i t)$ expansion, hence corresponds to a large volume expansion. The asymptotic expansion for $\lambda \rightarrow \infty$ however is expressed in terms of $t$ and is moreover manifestly \emph{polynomial!}  in $t$ for every order in $\frac{1}{\check{\lambda}}$. This suggests that $\lambda$ and $t$ are not entirely independent expansion variables but should perhaps be thought to correspond to certain phases of the combined problem in $\lambda$ and $t$, in analogy to the phases of \cite{Witten:1993yc}.
\item Although the volume of the resolved conifold is infinite since it is a non-compact CY, it is interesting to observe that the $\frac{1}{\lambda^2}$ term in the expansion comes with a factor $\frac{1}{3!}t^3$, which is the classical CY volume expressed by the special geometry. It is then natural to speculate that this term corresponds to the NS5 brane contribution, to obtain a negative sign, a factor of $1/(2 \pi i)^3$ should be included in the volume identification.
\end{itemize}
\end{rem}
\subsubsection{Strong and weak coupling resummations}\label{sec:resummation}
For the resolved conifold we analyze the non-perturbative content of the solution of the difference equation. We obtain the following:
\begin{prop}\label{conresumprop} $F_{np}(\lambda,t)$ can be expressed as:
\begin{equation}
F_{np}(\lambda,t) = \sum_{k=1}^{\infty} \frac{q^k}{k(2\sin k\lambda/2)^2}  - \frac{\partial}{\partial \check{\lambda}} \left( \frac{\check{\lambda}}{4\pi}  \sum_{k=1}^{\infty} \frac{e^{2\pi i k (t-1/2)/\check{\lambda}}}{k^2 \sin (\pi k /\check{\lambda})} \right)\,, \quad \check{\lambda} =\frac{\lambda}{2\pi}\,.
\end{equation}
Moreover, this expression can be written as:
\begin{equation}\label{GVNS}
F_{np}(\lambda,t)= F_{GV}(\lambda,t) -  \frac{\partial}{\partial \check{\lambda}} \left(\check{\lambda} \, F_{NS} (1/\check{\lambda},(t-1/2)/\check{\lambda})\right)\,,
\end{equation}
where 
\begin{equation}
F_{GV}(\lambda,t)= \sum_{k=1}^{\infty} \frac{q^k}{k(2\sin k\lambda/2)^2}\,,
\end{equation}
is the Gopakumar-Vafa resummation of \eqref{GVresum} for the resolved conifold and
\begin{equation}
F_{NS}(\hbar,t) = \frac{1}{4\pi} \sum_{k=1}^{\infty} \frac{e^{2\pi i k t}}{k^2 \sin (\pi k \hbar)}\,,
\end{equation}
is the refined topological string free energy for the resolved conifold in the Nekrasov-Shatashvili limit as determined in \cite{Hatsuda:2015oaa,Hatsuda:2015owa} following \cite{Hatsuda:2013oxa}.
\end{prop}
\begin{proof}
To prove this, we use the integral representation of the function $G_3$, discussed in \cite[Prop. 4.2]{BridgelandCon}, based on \cite[Prop. 2]{Narukawa} obtaining:
\begin{equation}
F_{np}(\lambda,t) = - \int_{\mathcal{C}} \frac{e^{(t+\check{\lambda})s}}{(e^s-1)(e^{\check{\lambda}s}-1)^2 } \frac{ds}{s}\,,
\end{equation}
which is valid for $\textrm{Re}\check{\lambda} >0$ and $-\textrm{Re}\check{\lambda} < \textrm{Re} t < \textrm{Re} (\check{\lambda}+1)$ and the contour $\mathcal{C}$ is following the real axis from $-\infty$ to $\infty$ avoiding $0$ by a small detour in the upper half plane.

To find the series expression in $\lambda$ including the perturbative and non-perturbative pieces from the integral representation we close the contour in the upper half plane and analyze the residues. In the upper half plane without zero, the integrand has two infinite sets of poles given by:
\begin{equation}
s= 2\pi i k \,, \quad \textrm{and} \quad s=\frac{2\pi i k }{\hat{\lambda}}\, \quad k \in \mathbb{N}\setminus\{0\},
\end{equation}
the first set corresponds to simple poles and the second set corresponds to double poles. To avoid higher order poles we may assume that either $\textrm{Im} \lambda \ne 0$ or that $\check{\lambda} \notin \mathbb{Q}$.\footnote{We note that there is nothing wrong with $\hat{\lambda}\in \mathbb{Q}$, it just requires a separate analysis of the poles of order three.}

\usetikzlibrary{calc,decorations.markings}
\begin{figure}
\begin{centering}
\begin{tikzpicture}
\draw (0,-0.5) -- (0,5.5);  
\draw (-5.5,0) -- (5.5,0);   
  \foreach \y in {0,...,4} {
  \node at (0,\y) {$\times$};
  };
  \foreach \y in {1,...,4} {
  \node at (-0.6,\y) {\y $ \cdot \, 2\pi i$  };
    };
  \foreach \y in {1,...,5} {
  \node at (2*\y/5,4*\y/5) {$\times$};
  };
   \foreach \y in {1,...,4} {
   \node at (2*\y/5+1,4*\y/5) {\y $ \cdot \, 2\pi i/\check{\lambda}$  };
    };



\draw[thick,red,xshift=0pt,
decoration={ markings,  
      mark=at position 0.2 with {\arrow{latex}}, 
      mark=at position 0.6 with {\arrow{latex}},
      mark=at position 0.8 with {\arrow{latex}}, 
      mark=at position 0.98 with {\arrow{latex}}}, 
      postaction={decorate}]
  (-5,0) -- (-0.5,0) arc (180:0:.5) -- (5,0);
\draw[thick,red,xshift=0pt,
decoration={ markings,
      mark=at position 0.2 with {\arrow{latex}}, 
      mark=at position 0.4 with {\arrow{latex}},
      mark=at position 0.6 with {\arrow{latex}}, 
      mark=at position 0.8 with {\arrow{latex}}}, 
      postaction={decorate}]
 (5,0) arc (0:180:5) -- (-5,0);
\end{tikzpicture}

\end{centering}
\caption{Illustration of the simple and double poles as well as the contour $\mathcal{C}$.}
\end{figure}
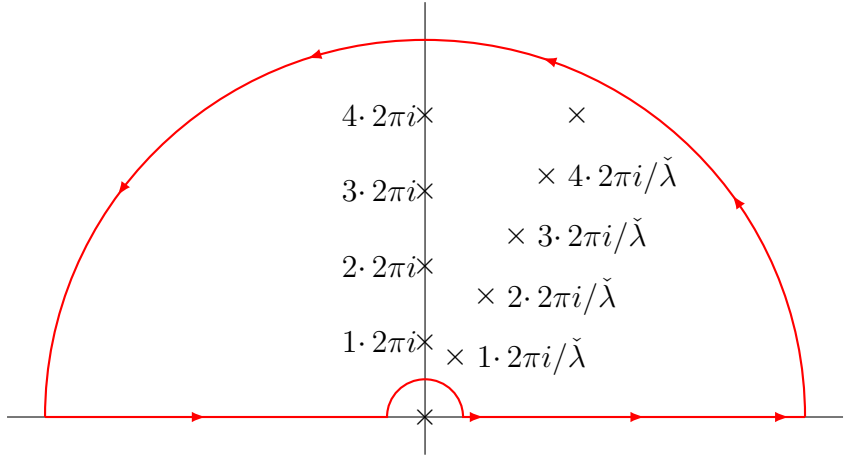

The contribution of the simple poles gives the first factor of the r.h.s. of the proposition. We denote the contribution of the double poles by $I_{db}$ and obtain from Cauchy's generalized integral formula, introducing $s'=\check{\lambda}\, s$:
\begin{equation}
I_{db}= 2\pi i \cdot \sum_{k=1}^{\infty} \frac{d}{ds'} f(s')|_{s'=2\pi i k}\,,
\end{equation}
where $f(s')= -\frac{e^{ts'/\check{\lambda}}}{s'(e^{s'/\check{\lambda}}-1)}$, which gives:
\begin{equation}
I_{db}= -\frac{1}{4\pi} \sum_{k=1}^{\infty} \frac{(1-2\pi i k t/\check{\lambda}) e^{2\pi i k (t-1/2)/\check{\lambda}}}{k^2 \sin (\pi k/\check{\lambda})} -\frac{1}{\check{\lambda}}   \sum_{k=1}^{\infty} \frac{ e^{2\pi i k t/\check{\lambda}}}{k (2\sin (\pi k/\check{\lambda}))^2}\,,
\end{equation}
the reformulation of the result follows after some substitution algebra.
\end{proof}

\begin{rem}
\begin{itemize}
\item The result \eqref{GVNS} matches the result of eq.~(5.6) of \cite{Hatsuda:2015owa} which was obtained by a generalized Borel resummation. There are differences of various relative factors of $2\pi i$ compared to  \cite{Hatsuda:2015owa}. These originate from our definition of $q=e^{2\pi i t}$ (opposed to $q=e^{-t}$, which is often used). Including the $2\pi i$ factors corresponds to the natural choice of variable giving convergence of $q$ series at large volume $\textrm{Im} t \rightarrow \infty$ as well as having the periodicity in shifts of the $B-$field $t\rightarrow t+1$. Indeed, this makes the $1/2$ unit of $B$-field shift in the argument of $F_{NS}$, which was discussed in \cite{Hatsuda:2013oxa,Hatsuda:2015oaa,Hatsuda:2015owa} very clear.
\item Since this computation reproduces results which agree with \cite{Hatsuda:2013oxa}, one may consult \cite[Sec. 4.2]{Hatsuda:2013oxa} for comments and comparison to the proposal of \cite{Lockhart:2012vp}. 
\item A similar computation using integral representations occurring in Chern-Simons theory was done in \cite{Krefl:2015vna}. 
\end{itemize}
\end{rem}
We conclude this subsection by noting that the special function determined by the difference equation does indeed contain the non-perturbative structure of the resolved conifold which was obtained by various different methods in \cite{Hatsuda:2015oaa,Hatsuda:2015owa,Krefl:2015vna}. This in particular confirms the expectation of \cite{BridgelandCon}, that the Tau function of the Riemann-Hilbert problem associated to wall-crossing of DT invariants of the resolved conifold defines a non-perturbative topological string partition function.


\subsection{Non-perturbative conifold and non-commutative DT}

With the non-perturbative topological string partition function at hand one can revisit the relation between topological strings and the generating function of non-commutative Donaldson-Thomas invariants studied in \cite{Szendroi}. We introduce therefore:\footnote{Compared to \cite{Szendroi}, $q_\lambda$ has a different sign and $2\pi i$ factors are added to $t$.}
\begin{equation}\label{ncdt}
Z_{\textrm{NCDT}}(\lambda,t):= M(q_{\lambda})^2 \, \prod_{k=1}^{\infty} \left(1-  e^{-2\pi i t} q_{\lambda}^k\right)^k \prod_{k=1}^{\infty} \left(1-  e^{2\pi i t} q_{\lambda}^k\right)^k\,,
\end{equation}
where the MacMahon function is given by:
\begin{equation}
M(q)= \prod_{k=1}^{\infty} (1-q^k)^{-k}\,,
\end{equation}
and $q_{\lambda}=e^{i\lambda}\,.$
Defining the non-perturbative topological string free partition function:
\begin{equation}\label{nonppart}
Z_{np}(\lambda,t):= \exp (F_{np}(\lambda,t))\,,
\end{equation}
as the exponential of \eqref{resconfreedef}, we obtain the following:
\begin{prop}
The relation between the generating function of non-commutative DT invariants \eqref{ncdt} and the non-perturbative topological string partition function \eqref{nonppart} is given by:
\begin{equation}
Z_{\textrm{NCDT}}(\lambda,t) =  M(q_{\lambda})^2\,\cdot Z_{np}(\lambda,t+1)\, \cdot Z_{np}(-\lambda,-t)\,,
\end{equation}
\end{prop}
\begin{proof}
We consider the following property, proved in \cite[Prop. 4.3]{BridgelandCon} for the function $G_3$ defined in \eqref{g3def}:
\begin{equation}
G_{3}(z+\omega_2\, |\, \omega_1,\omega_2) \cdot G_3(z\, | \, \omega_1,-\omega_2)= \prod_{k=1}^{\infty} (1-x_2\, q_2^k)^k\cdot  \prod_{k=1}^{\infty} (1-x_2^{-1}\, q_2^k)^k\,,
\end{equation}
valid for $\textrm{Im}(\omega_1/\omega_2)>0$ and $z\in \mathbb{C}$, where:
$$x_2=\exp (2\pi i z/\omega_2)\,, \quad q_2=\exp (2\pi i \omega_1/\omega_2)\,. $$
The claim of the proposition follows by considering the property for $G_{3}(t\, | \check{\lambda},1)$ which defined $F_{np}$ and further noting that the function $G_3$ is invariant under simultaneous rescaling of all three arguments, see e.~g.~\cite[Prop. 4.2]{BridgelandCon}.
\end{proof}
\begin{rem}
We remark that relations between $Z_{\textrm{NCDT}}$ and the partition function for the resolved conifold obtained by exponentiating the GV resummation were already observed in \cite{Szendroi} and interpreted physically in e.~g.~\cite{Jafferis:2008uf,Aganagic:2009kf}. The result obtained here is however exact and non-perturbative and therefore perhaps a little more surprising since proposition \ref{sec:resummation} shows that the non-perturbative free energy contains more information than the GV resummation.
\end{rem}


\subsection{Conifold locus of the resolved conifold}
A natural question to ask is about the precise relation between the free energies of the resolved and deformed conifold geometries. Relatedly, the free energy of the resolved conifold should be the easiest example to explicitly test the universal behavior of topological string theory near conifold type singularities. Since the resolved conifold corresponds to the small resolution of the conifold singularity, the singular locus corresponds to the locus where $t$, the K\"ahler parameter of the $\mathbb{P}^1$ shrinks to zero $t\rightarrow 0$. Since the free energy of the resolved conifold \eqref{resconfree} is given as a series expansion in $q=e^{2\pi i t}$, which does not converge as $t\rightarrow 0$ one must  analytically continue this result. We will prove the following:
\begin{thm}\label{conanalytic}
For $g >1$ the analytic continuation of the genus $g$ free energy of the resolved conifold to the regime $t\rightarrow 0$ is given by:

\begin{equation}
\begin{split}
\tilde{F}^g(t)&= \frac{B_{2g}}{2g(2g-2)} \frac{1}{(2\pi t)^{2g-2}} + (-1)^g \frac{B_{2g} B_{2g-2}}{2g (2g-2)(2g-2)!} \\
 &+ \frac{B_{2g}}{2g(2g-2)!} \sum_{d=g}^{\infty} (-1)^{d+1} \frac{B_{2d}}{2d(2d-2g+2)!} (2\pi t)^{2d-2g+2}\,.
\end{split}
\end{equation}
\end{thm}

\begin{proof}
We begin by representing $\tilde{F}^g(t)$ for the resolved conifold as in \eqref{derrep}:
\begin{equation}
\tilde{F}^g(t)=\frac{(-1)^{g-1}B_{2g}}{2g (2g-2)!} \,\theta_q^{2g-3} \textrm{Li}_0(q)\,, \quad g\ge1\,,
\end{equation}
starting from $\textrm{Li}_0(q)=\frac{q}{(1-q)}$ and using the property:
\begin{equation} \label{polylogder}
\theta_q \textrm{Li}_s(q) =\textrm{Li}_{s-1} (q)\,,  \quad \theta_q:= q \,\frac{d}{dq}\,.
\end{equation}
We introduce $a=2\pi i t$, the expression becomes:
\begin{equation}
\begin{split}
\tilde{F}^g(a)&=\frac{(-1)^{g}B_{2g}}{2g (2g-2)!} \,\partial_a^{2g-3} \left(\frac{1}{a} \left( \frac{a}{e^a-1}\right)\right) \,\\&=\frac{(-1)^{g}B_{2g}}{2g (2g-2)!} \,\partial_a^{2g-3}  \left( \sum_{n=0}^{\infty} B_n \frac{a^{n-1}}{n!}\right)\\
&= \frac{(-1)^g B_{2g}}{2g(2g-2) a^{2g-2}} + (-1)^g \frac{B_{2g} B_{2g-2}}{2g (2g-2) (2g-2)!} \\&+ (-1)^{g-1} \frac{B_{2g}}{2g(2g-2)!} \sum_{d=g}^{\infty} \frac{B_{2d}}{2d(2d-2g+2)!} a^{2d-2g+2}\,,
\end{split}
\end{equation}
where in the first and second line the expression $Li_0(q)$ was brought into the form of the generating function of Bernoulli numbers. The result on the third line follows from considering separately the differential operator acting on the singular piece $1/a$ of the r.h.s., the constant piece comes from the differential operator acting on the term  with power $a^{2g-3}$, the rest is collecting all the higher powers. We have furthermore used the vanishing of $B_{2n+1}, n>0$. The statement of the theorem follows by substituting back $a=2\pi i t$.

\end{proof}

\begin{rem}
\begin{itemize}
\item This result is a mathematical proof of the gap condition used in \cite{Huang:2006si,Huang:2006hq}, based on the physical expectation outlined in sec.~\ref{sec:boundary}, as a condition imposed on the local behavior of topological string theory on arbitrary families of CY manifolds. In the case at hand the result does not serve any computational purpose since the full exact and non-perturbative expression for topological string theory was already given, it is nevertheless gratifying to see the conjectured behavior to hold mathematically rigorously in this example.
\item We think the result should have an interesting interpretation in enumerative geometry. The leading singular term gives the Euler characteristic of genus $g$ Riemann surfaces proved in \cite{HZ}, in the constant term, the contributions from constant maps determined in \cite{Faber} make a somewhat surprising reappearance, although we had explicitly not included these from the start. It is therefore natural to expect also the higher degree terms to have an interpretation. This may open up new paths to study Gromov-Witten theory around the conifold point of arbitrary families of threefolds. See also \cite{Szendroi} for speculation about the enumerative geometry content of the locus $t\rightarrow 0$ at strong coupling $\lambda\rightarrow i\infty$.

\item An analytic continuation of the free energies for the resolved conifold for $g>1$ was obtained in \cite[Sec. 4.3] {Pasquetti:2009jg}using the following series expansion of the polylogarithm:
\begin{equation}
    \textrm{Li}_{s}(e^{2\pi i t})= \Gamma(1-s) \, \sum_{m\in \mathbb{Z}} 2\pi i(m-t)^{s-1}\,,
\end{equation}
valid for $\textrm{Re}(s)<0$ and $t\notin \mathbb{Z}$. This gives for the free energies \cite{Pasquetti:2009jg}:
\begin{equation}\label{masslessd0}
\tilde{F}^g(t)= \frac{B_{2g}}{2g (2g-2)} \sum_{k\in \mathbb{Z}} \frac{1}{(2\pi(k-t))^{2g-2}}\,, \quad g>1\,.
\end{equation}

\item From a physical perspective, the singular behavior signals a hypermultiplet becoming massless in the effective four dimensional theory in the limit $t\rightarrow 0$. It is known that the resolved conifold supports only one BPS state from the M-theory perspective, this corresponds to the only non-vanishing GV invariant $n^0_1=1$. However from the $4d$ perspective there is an infinite set of BPS states supported on this geometry corresponding to the Kaluza-Klein modes of the $M2$ brane, these correspond to the $D_0$ brane charge bound to the $D_2$ brane. 
While the statement of the theorem \ref{conanalytic} only shows the singular behavior as the pure $D_2$ becomes massless at $t\rightarrow 0$, the analytic continuation \eqref{masslessd0} obtained in \cite{Pasquetti:2009jg} shows the whole tower of $D_2$ bound to an integer amount of $D_0$ becoming massless at the loci $t\in\mathbb{Z}$.

\end{itemize}
\end{rem}

To fully test the universal behavior of topological string theory near a conifold singularity in this example, we need to analyze furthermore the behavior of the genus $0$ and genus $1$ free energies. 

We obtain the following:
\begin{prop}
The analytic continuation of the genus 0 free energy (the prepotential) $\tilde{F}^0(t)=\textrm{Li}_3(q)$ to the region $t\rightarrow 0$ is given by:
\begin{equation}
\tilde{F}^0(t)= (2\pi)^2 \left( \frac{t^2}{2} \log (2\pi i t) -\frac{3}{4}t^2\right) -(2\pi i)^3 \frac{t^3}{12} - \sum_{n=1}^{\infty} \frac{B_{2n}}{2n(2n+2)!} (-1)^{n+1} (2\pi t)^{2n+2}\,.
\end{equation}
\end{prop}
\begin{proof}
Since the prepotential $\tilde{F}^0(t)$ is a locally defined function, we need to analytically continue it starting from functions which are geometrically globally defined. At genus zero such an object is the Yukawa coupling\footnote{This definition using only the $q$ expansion part of $F^0(t)$ misses some classical pieces which would correspond to the classical triple intersection numbers.}
$$ C_{ttt}:=\frac{1}{(2\pi i)^3}\frac{\partial^3}{\partial t^3} \tilde{F}^0(t)=\theta^3_q \textrm{Li}_3(q)=\frac{q}{1-q}\,, $$ 
introducing again $a=2\pi i t$, we thus have:
\begin{equation}
C_{ttt}= \partial_a^3 \tilde{F}^0(a)= -1 -\frac{1}{a}\left(\frac{a}{q-1}\right) = -1 - \sum_{n=0}^{\infty}B_n \frac{a^{n-1}}{n!}\,,
\end{equation}
which can be integrated to give:
\begin{equation}
\tilde{F}^0(a)= -\left( \frac{a^2}{2} \log a -\frac{3}{4}a^2\right) -\frac{a^3}{12} - \sum_{n=1}^{\infty} \frac{B_{2n}}{2n(2n+2)!} a^{2n+2}\,.
\end{equation}
The statement of the proposition follows by substituting $a=2\pi i t$.
\end{proof}

For the genus $1$ free energy we have:
\begin{equation}
\tilde{F}^1(t)= -\frac{1}{12} \log (1-q)= -\frac{1}{12} \log ( 2\pi i t) -\frac{1}{12} \log\left(1+\sum_{n=2}^{\infty} \frac{(2\pi i t)^{n-1}}{n!}\right)\,, 
\end{equation}
which shows that the assumptions of corollary \ref{diffarb} are met. We can now consider $\Lambda \in \mathbb{C}^*$ and the rescaling:
$$ \lambda' = \lambda \cdot \Lambda\,, \quad  t' = t \cdot \Lambda\,,$$
and define:
$$ F'(\lambda',t') := \lim_{\Lambda \rightarrow \infty} \left( F(\lambda',t) + \frac{t^2}{2}\log \left(\frac{\Lambda}{2\pi i}\right) -\frac{1}{12}\log \left(\frac{\Lambda}{2\pi i}\right)\right)\,, $$
one fine
$$ F'(\lambda',t')=F_{\textrm{def}}(\check{\lambda'},t')\,, \quad \check{\lambda}'=\frac{\lambda}{2\pi}\,,$$
where $F_{\textrm{def}}$ is the free energy of the deformed conifold \ref{defconfree}. The different normalization of $\lambda$ is also the reason for the $2\pi i$ factors on the cutoff $\Lambda$.
\section{Conclusions}\label{sec:conclusions}

In this work, we addressed an intrinsic characterization of non-perturbative topological string theory. Hereby the topological string free energy is a solution to a difference equation mixing the moduli and the string coupling. The difference equations are derived using the knowledge of the asymptotic expansion and bypass any further physical or mathematical dualities. The difference equations as well as their solutions were obtained for the deformed and resolved conifold geometries. The strong coupling expansion as well as the non-perturbative content were obtained in these cases and matched to known and expected results in the literature. Fortunately, the deformed conifold also gives the expected universal behavior of topological string theory on any family of Calabi-Yau threefolds near a singularity where finitely many states of the effective field theory in $4d$ become massless. Both the difference equation as well as its analytic solution in the string coupling can thus be obtained universally in a limit where the coordinate giving the mass of the particles becoming massless at the singularity as well as the topological string coupling are rescaled and the rescaling is sent to $\infty$.

Although the full exact solution is lost in the rescaled limit, the remaining qualitative results are universal and should shed new lights on contexts where topological string theory connects to problems of quantum gravity and thus requires the study of topological string theory on families of compact CY manifolds. One such problem is the connection to black hole partition functions \cite{Ooguri:2004zv}. It would for example be interesting to match the universal non-perturbative structure which also holds for compact CY to computations of the black hole partition function, perhaps along the lines of \cite{Sen:2011ba}.\footnote{I would like to thank Gabriel Lopes Cardoso for pointing this out, commenting on the potential use of \cite{Alim:2015qma} in the study of black holes.} Another quantum gravity context, where the precise knowledge of the geometry of CY moduli spaces and their non-perturbative content is important is the swampland program which has attracted a lot of attention recently. We expect the quantitative handle on the universal non-perturbative structure of topological strings to be relevant for example in the context of the swampland distance conjecture \cite{Ooguri:2006in}. Relatedly, one may also expect the non-perturbative structure of topological string theory to be relevant in the study of non-perturbative corrections to quaternionic K\"ahler geometries \cite{Ooguri:1996me,Alexandrov:2011va,Alexandrov:2013yva}.

We note that difference equations provide the natural arena to study the connections between topological strings and integrable hierarchies. This was indeed the context for \cite{ADKMV}. The difference equations there were obtained from the quantization of the mirror curves, which are expressed in terms of $\mathbb{C}^*$ variables. The difference equations and their solutions in this context connect however naturally to open topological string theory \cite{ADKMV,Aganagic:2011mi}. The difference equations studied in this work are equations in the closed string moduli, they are closer in spirit to the difference equation conjectured for instance in \cite{Pandharipande} for the Gromov-Witten potential of $\mathbb{P}^1$ and which provides the link to the Toda integrable hierarchy. Indeed, using the difference equation for the resolved conifold, a conjecture \cite{Brini} relating Gromov-Witten theory of the resolved conifold to the Ablowitz-Ladik integrable hierarchy was proved in \cite{alim2021integrable}. It is tempting to speculate that there should be an underlying quantum mechanical problem that gives rise to the difference equations studied here. The latter should be a quantization problem related to the moduli space of closed strings of the geometry, which may close the circle of ideas by relating it to the quantum mechanical interpretation of the closed string as a wave-function \cite{Witten:1993ed}. Recently Quantum K-theory was studied in e.~g.~\cite{Jockers:2019wjh} as a quantum deformation problem which can be applied to the Picard-Fuchs operators of any Calabi-Yau geometry, it would be interesting to understand potential links to the approach in this paper. A further context where the difference equations studied in this work appear is as the equations characterizing the perturbative part of the Nekrasov-Okounkov partition functions of $\mathcal{N}=2$ gauge theories, see \cite[Appendix A]{NO}. The current work suggests that in addition to the perturbative piece, the difference equations also fix the non-perturbative content. It would be interesting to study the implications of this further both in the gauge theory as well as in more general contexts.


\subsection*{Acknowledgements}
I would like to thank Arpan Saha for collaboration on related projects as well as Florian Beck and Peter Mayr for comments on the draft. I have benefited from many discussions with members of the Emmy Noether research group on String Mathematics as well as from discussions with Vicente Cort\'es, J\"org Teschner, Iv\'an Tulli, Timo Weigand and Alexander Westphal on related projects within the quantum universe cluster of excellence. This work is supported through the DFG Emmy Noether grant AL 1407/2-1.

\begin{appendix}
\section{Proof of the difference equation}\label{sec:proof}
We provide here the proof of \cite{Iwaki,alim2020difference} of the difference equations for the free energies of the deformed and the resolved conifold.
\begin{proof}
The proof of the theorems relies on the following.
Consider the generating function of Bernoulli numbers:
\begin{equation}
\frac{w}{e^w-1} = \sum_{n=0}^{\infty} B_n \frac{w^n}{n!}\,.
\end{equation}
Applying $w \frac{d}{dw}$ to both sides and rearranging gives:
\begin{equation}
\frac{w^2 e^w}{(e^w-1)^2} = B_0 - \sum_{n=2}^{\infty} \frac{B_n}{n (n-2)!} w^n = 1 -\sum_{g=1}^{\infty} \frac{B_{2g}}{2g (2g-2)!} w^{2g}\,,
\end{equation}
where the last equality is obtained by noting that all $B_{2n+1}, n\in \mathbb{N}\setminus\{0\}$ vanish. This yields the following:
\begin{equation}
(e^w-2+e^{-w})  \left(\frac{1}{w^2} -\sum_{g=1}^{\infty} \frac{B_{2g}}{2g (2g-2)!} w^{2g-2} \right) = 1\,.
\end{equation}

In the next step, we replace on both sides of this equation the variable $w$ by an operator acting on functions of $t$ namely:
$$ w \rightarrow \lambda \frac{\partial}{\partial a}$$
For the deformed conifold, we act with both sides on $\log a$, obtaining:
\begin{equation}
(e^{\lambda \partial_a}-2\cdot \textrm{id}+e^{-\lambda\partial_a})  \left( \lambda^{-2} \partial_a^{-2} -\sum_{g=1}^{\infty} \frac{B_{2g}}{2g (2g-2)!}  \lambda^{2g-2}\,\partial_a^{2g-2} \right) \log a = \textrm{id}\cdot \log a\,,
\end{equation}
we use 
$$ \partial_a^{2g-2}  \log a= - (2g-3)! \,a^{2-2g},$$ 
and interpret $\partial_a^{-1}$ as an anti-derivative to obtain:
\begin{equation}
(e^{\lambda \partial_a}-2\cdot \textrm{id}+e^{-\lambda\partial_a})  \left( F^0(a) -\frac{1}{12}\log a +\sum_{g=2}^{\infty} \frac{B_{2g}}{2g (2g-2)}  \frac{\lambda^{2g-2}}{a^{2g-2}} \right)  =  \log a= \frac{\partial^2}{\partial a^2} F^0(a)\,.
\end{equation}

For the resolved conifold we start from $\textrm{Li}_1(q)=-\log(1-q)$ and use the property:
\begin{equation} \label{polylogder}
\theta_q \textrm{Li}_s(q) =\textrm{Li}_{s-1} (q)\,,  \quad \theta_q:= q \,\frac{d}{dq}\,,
\end{equation}
we write 
\begin{equation}\label{derrep}
\tilde{F}^g=\frac{(-1)^{g-1}B_{2g}}{2g (2g-2)!} \,\theta_q^{2g-2} \textrm{Li}_1(q)\,, \quad g\ge1\,.
\end{equation}
We make the replacement:
$$ w \rightarrow \check{\lambda} \frac{\partial}{\partial t}= i \lambda \theta_q\,.$$

Acting with both sides on $\textrm{Li}_1(q)$ we obtain:
\begin{equation}
(e^{\check{\lambda} \partial_t}-2\cdot \textrm{id}+e^{-\check{\lambda}\partial_t})  \left(- \lambda^{-2} \theta_q^{-2} -\sum_{g=1}^{\infty} \frac{(-1)^{g-1}B_{2g}}{2g (2g-2)!}  \lambda^{2g-2}\,\theta_q^{2g-2} \right) \textrm{Li}_1(q) = \textrm{id}\cdot \textrm{Li}_1(q) \,,
\end{equation}
by using $\theta_q^{2} \tilde{F}^0=  \theta_q^{2} \textrm{Li}_3(q)= \textrm{Li}_1(q) $ and interpreting $\theta_q^{-1}$ as an anti-derivative, we obtain:
\begin{equation}
(e^{\check{\lambda} \partial_t}-2\cdot \textrm{id}+e^{-\check{\lambda}\partial_t})  \tilde{F}(\lambda,t) = -\theta_q^2 \tilde{F}^0(q) \,,
\end{equation}
which proves the theorem.
\end{proof}

The following corollary was proved in \cite{alim2020difference} but also holds for the deformed conifold case so we inlude it here as well:
\begin{cor}
For every $g\ge 1$ the difference equation gives a recursive differential equation which determines $\frac{\partial^2}{\partial t^2}F^g(t)\, \, g \ge 1$ by:
\begin{equation}
\sum_{k=0}^{g} \frac{1}{(2g-2k+2)!} \left(\frac{1}{2\pi} \frac{\partial}{\partial t}\right)^{2g-2k+2} F^k(t) =0 \,, \quad g\ge1\,.
\end{equation}
\end{cor}
\begin{proof}
This follows from expanding the L.H.S. of the theorem in $\lambda$ and then matching the coefficients of $\lambda^{2g}$ on both sides.
\end{proof}

\section{The quintic and its mirror}\label{sec:quintic}
In the following the quintic and its mirror are reviewed. The exposition follows \cite{Alimlectures}, more details can be found in \cite{Candelas:1990rm,Cox:1999ms}.   
The quintic $X$ denotes the CY manifold defined by 
\begin{equation}
X:=\{P(x)=0\} \subset \mathbb{P}^4\,,
\end{equation} 
where $P$ is a homogeneous polynomial of degree 5 in 5 variables $x_1, \dots ,x_5$. The mirror quintic $\check{X}$ can be constructed using the Greene-Plesser construction \cite{Greene:1990ud}. Equivalently it may be constructed using Batyrev's dual polyhedra \cite{Batyrev:1993dm} in the toric geometry language\footnote{For a review of toric geometry see Refs.~\cite{Greene:1996cy,Hori:2003ms}.}.
In the Greene-Plesser construction the family of mirror quintics is the one parameter family of quintics defined by
\begin{equation}\label{GP}
 \{ p(\check{X})=\sum_{i=1}^5 x_i^5-\psi \prod_{i=1}^5 x_i=0 \}  \in \mathbbm{P}^4\,,
\end{equation}
after a $(\mathbbm{Z}_5)^3$ quotient and resolving the singularities.

In the following, the mirror construction following Batyrev will be outlined. The mirror pair of CY 3-folds $(X,\check{X})$ is given as hypersurfaces in toric ambient spaces $(W,\check{W})$. The mirror symmetry construction of Ref.~\cite{Batyrev:1993dm} associates the pair $(X,\check{X})$ to a pair of integral reflexive polyhedra $(\Delta,\check{Delta})$. 

\subsection{\it The A-model geometry}
The polyhedron $\Delta$ is characterized by $k$ relevant integral points $\nu_i$ lying in a hyperplane of distance one from the origin in $\mathbbm{Z}^5$, $\nu_0$ will denote the origin following the conventions of Refs. \cite{Batyrev:1993dm,Hosono:1993qy}.  The $k$ integral points $\nu_i(\Delta)$ of the polyhedron $\Delta$ correspond to homogeneous coordinates $u_i$ on the toric ambient space $W$ and satisfy $n=h^{1,1}(X)$ linear relations:
\begin{equation}\label{toricrel}
\sum_{i=0}^{k-1} l_i^a \, \nu_i=0\, , \quad a=1,\dots,n\,.
\end{equation}
The integral entries of the vectors $l^a$ for fixed $a$ define the weights $l_i^a$ of the coordinates $u_i$ under the $\mathbbm{C}^*$ actions
$$ u_i \rightarrow (\lambda_a)^{l_i^a} u_i\,, \quad \lambda_a \in \mathbbm{C}^*\,.$$

The $l_i^a$ can also be understood as the $U(1)_a$ charges of the fields of the gauged linear sigma model (GLSM) construction associated with the toric variety \cite{Witten:1993yc}. The toric variety $W$ is defined as $W\simeq (\mathbbm{C}^{k}-\Xi)/(\mathbbm{C}^*)^n$, where $\Xi$ corresponds to an exceptional subset of degenerate orbits. To construct compact hypersurfaces, $W$ is taken to be the total space of the anti-canonical bundle over a compact toric variety. The compact manifold $X \subset W$ is defined by introducing a superpotential $\mathcal{W}_X=u_0 p(u_i)$ in the GLSM, where $x_0$ is the coordinate on the fiber and $p(u_i)$ a polynomial in the $u_{{i>0}}$ of degrees $-l_0^a$. At large K\"ahler volumes, the critical locus is at $u_0=p(u_i)=0$ \cite{Witten:1993yc}. 

The quintic is the compact geometry given by a section of the anti-canonical bundle over  $\mathbbm{P}^4$. The charge vectors for this geometry are given by:
\begin{equation}\label{chargevec}
\begin{array}{ccccccc}
&u_0&u_1&u_2&u_3&u_4&u_5\\
l=&(-5&1&1&1&1&1\, )\,.
\end{array}
\end{equation}
The vertices of the polyhedron $\Delta$ are given by:
\begin{eqnarray}
&&\nu_0=(0,0,0,0,0)\,,\quad \nu_1=(1,0,0,0,0)\,,\quad \nu_2=(0,1,0,0,0)\,,\nonumber\\
&&\nu_3=(0,0,1,0,0)\,,\quad \nu_4=(0,0,0,1,0)\,,\quad \nu_5=(-1,-1,-1,-1,0)\,.
\end{eqnarray}

\subsection{\it The B-model geometry}
The B-model geometry $\check{X}\subset \check{W}$ is determined by the mirror symmetry construction of Refs.~\cite{Hori:2000kt,Batyrev:1993dm} as the vanishing locus of the equation
\begin{equation}
p(\check{X})=\sum_{i=0}^{k-1} a_i y_i =\sum_{\nu_i\in \Delta} a_i X^{\nu_i}\, ,
\end{equation}
where $a_i$ parameterize the complex structure of $\check{X}$, $y_i$ are homogeneous coordinates \cite{Hori:2000kt} on $\check{W}$ and $X_m\, , m=1,\dots,4$ are inhomogeneous coordinates on an open torus $(\mathbbm{C}^*)^4 \subset \check{W}$  and $X^{\nu_i}:=\prod_m X_m^{\nu_{i,m}} $ \cite{Batyrev:1993wa}. The relations (\ref{toricrel}) impose the following relations on the homogeneous coordinates
\begin{equation}
\prod_{i=0}^{k-1} y_i^{l_i^a}=1\, ,\quad a=1,\dots,n=h^{2,1}(\check{X})=h^{1,1}(X)\, .
\end{equation}
The important quantity in the B-model is the holomorphic $(3,0)$ form which is given by:
\begin{equation}\label{defomega0}
\Omega(a_i)= \textrm{Res}_{p=0} \frac{1}{p(\check{X})} \prod_{i=1}^4 \frac{dX_i}{X_i} \, .
\end{equation}
Its periods 
\begin{equation}
\pi_{\alpha}(a_i)=\int_{\gamma^\alpha} \Omega(a_i)\, , \quad \gamma^{\alpha} \in H_3(\check{X})\,,\quad\alpha=0,\dots, 2h^{2,1}+1\, ,
\end{equation} 
are annihilated by an extended system of GKZ  \cite{Gelfand:1989} differential operators
\begin{eqnarray}
&&\mathcal{L}(l)= \prod_{l_i >0} \left( \frac{\partial}{\partial a_i}\right)^{l_i} -\prod_{l_i<0} \left( \frac{\partial}{\partial a_i}\right)^{-l_i}\, ,\\
&&\mathcal{Z}_k =\sum_{i=0}^{k-1} \nu_{i,j} \theta_i\, , \quad j=1,\dots,4\, . \quad \mathcal{Z}_0 = \sum_{i=0}^{k-1} \theta_i +1\,,\quad \theta_i=a_i \frac{\partial}{\partial a_i}\,,
\end{eqnarray}
where $l$ can be a positive integral linear combination of the charge vectors $l^a$. The equation $\mathcal{L}(l)\, \pi_0(a_i)=0$ follows from the definition (\ref{defomega0}). The equations $\mathcal{Z}_k\,\pi_\alpha(a_i)=0$ express the invariance of the period integral under the torus action and imply that the period integrals only depend on special combinations of the parameters $a_i$
\begin{equation}\label{lcs}
\pi_\alpha(a_i) \sim \pi_\alpha(z_a)\, ,\quad z_a=(-)^{l_0^a} \prod_i a_i^{l_i^a}\, ,
\end{equation}
the $z_a\,, a=1,\dots,n$ define local coordinates on the moduli space $\mathcal{M}$ of complex structures of $\check{X}$.

The charge vector defining the A-model geometry in Eq.~(\ref{chargevec}) gives the mirror geometry defined by:
\begin{equation}\label{Batyrev}
 p(\check{X})=\sum_{i=0}^5 a_i y_i =0\, ,
\end{equation}
where the coordinates $y_i$ are subject to the relation
\begin{equation}
 y_1 y_2 y_3 y_4 y_5 = y_0^5\,.
\end{equation}
Changing the coordinates $y_i=x_i^5,\, i=1,\dots,5$ shows the equivalence of (\ref{GP}) and (\ref{Batyrev}) with 
\begin{equation}
 \psi^{-5}=-\frac{a_1 a_2 a_3 a_4 a_5}{a_0^5}=: z\,.
\end{equation}

Furthermore, the following Picard-Fuchs  (PF) operator annihilating $\tilde{\pi}_{\alpha}(z_i)=a_0\,\pi_{\alpha}(a_i)$ is found:
\begin{equation}\label{PF}
\mathcal{L}=\theta^4- 5z \prod_{i=1}^4 (5\theta+i)\,, \quad \theta=z \frac{d}{dz}\, .
\end{equation}
The discriminant of this operator is
\begin{equation}
  \label{eq:Discriminant}
\Delta=1-3125\,z\,.
\end{equation}
and the Yukawa coupling can be computed:
\begin{equation}
C_{zzz}=\frac{5}{z^3\, \Delta}\,.
\end{equation}

The PF operator gives a differential equation which has three regular singular points which correspond to points in the moduli space of the family of quintics where the defining equation becomes singular or acquires additional symmetries, these are the points:
\begin{itemize}
 \item $z=0$ , the quintic at this value corresponds to the quotient of $\prod_{i=1}^5 x_i=0$ which is the most degenerate Calabi-Yau and corresponds to large radius when translated to the A-model side.
\item $z=5^{-5}$ this corresponds to a discriminant locus of the differential equation (\ref{PF}) and also to the locus where the Jacobian of the defining equation vanishes. This type of singularity is called a \emph{conifold} singularity. 
\item $z=\infty$ , this is known as the Gepner point in the moduli space of the quintic and it corresponds to a non-singular CY threefold with a large automorphism group. This is reflected by a monodromy of order 5.
\end{itemize}

\end{appendix}

\providecommand{\href}[2]{#2}\begingroup\raggedright\endgroup


\end{document}